\documentclass[fleqn]{llncs} 

\usepackage{
amsmath,amssymb}
\usepackage[T1]{fontenc} 

\usepackage{tikz} 
\usepackage{verbatim}
\usetikzlibrary{calc,positioning,arrows}
\usepackage{paralist}
\usepackage[ruled,vlined] 
       {algorithm2e} 
\SetAlgoInsideSkip{medskip}

\usepackage{graphicx,xcolor}
\usepackage{float}
\usepackage{wrapfig}

\usepackage{hyperref}
\hypersetup{
  colorlinks   = true, 
  urlcolor     = blue, 
  linkcolor    = blue, 
  citecolor   = blue 
}

\newtheorem{observation}{Observation}

{\bf}{\it}
{\bf}{\it}

\newcommand{\NP}{{\sf NP}}

\newcommand{\svcfc}{\mathrm{svcfc}}
\newcommand{\dist}{\mathrm{dist}}
\newcommand{\diam}{\mathrm{diam}}
\newcommand{\rad}{\mathrm{rad}}
\newcommand{\ecc}{\mathrm{ecc}}

\def\3SAT{\textup{\textsc{3-sat}}}

\usepackage{lineno}

\tikzstyle{vertex}=[draw,circle,inner sep=1.3pt,fill=black]
\tikzstyle{clause}=[draw,circle,inner sep=1.3pt,fill=black]
\tikzstyle{var}=[draw,rectangle,inner sep=1.5pt,fill=black] 

\begin{document}

\title{
The complexity of strong conflict-free vertex-connection $k$-colorability
}
\author{Sun-Yuan Hsieh\inst{1} \and 
Hoang-Oanh Le\inst{2} \and 
Van Bang Le\inst{3} 
\and Sheng-Lung Peng\inst{4}
}
\institute{
Department of Computer Science and Information Engineering, National Cheng Kung University, No.~1, University Road, Tainan, 70101, Taiwan\\
\email{hsiehsy@mail.ncku.edu.tw}
\and
Independent Researcher, Berlin, Germany\\
\email{HoangOanhLe@outlook.com}
\and
Institut f\"{u}r Informatik, Universit\"{a}t Rostock,
Rostock, Germany\\ 
\email{van-bang.le@uni-rostock.de}
\and
Department of Creative Technologies and Product Design, National Taipei University of Business, Taoyuan, 32462, Taiwan\\
\email{slpeng@ntub.edu.tw}
}

\maketitle


\begin{abstract}
We study a new variant of graph coloring by adding a connectivity constraint. 
A path in a vertex-colored graph is called \emph{conflict-free} if there is a color that appears exactly once on its vertices. 
A connected graph $G$ is said to be \emph{strongly conflict-free vertex-connection $k$-colorable} if $G$ admits a vertex $k$-coloring such that any two distinct vertices of $G$ are connected by a  conflict-free \emph{shortest} path.

Among others, we show that deciding whether a given graph is strongly conflict-free vertex-connection $3$-colorable is \NP-complete even when restricted to $3$-colorable graphs with diameter~$3$, radius~$2$ and domination number~$3$, and, assuming the Exponential Time Hypothesis (ETH), cannot be solved in $2^{o(n)}$ time on such restricted input graphs with $n$ vertices. 
This hardness result is quite strong when compared to the ordinary $3$-{\sc coloring} problem: it is known that $3$-{\sc coloring} is solvable in polynomial time in graphs with bounded domination number, and assuming ETH, cannot be solved in $2^{o(\sqrt{n})}$ time in $n$-vertex graphs with diameter~$3$ and radius~$2$.
On the positive side, we point out that a strong conflict-free vertex-connection coloring with minimum color number of a given split graph or a co-bipartite graph can be computed in polynomial time. 

\end{abstract}

\setcounter{footnote}{0} 

\section{Introduction and results}
Graph connectivity is a fundamental topic in graph theory and combinatorial optimization. In the last years, a number of colored versions of graph connectivity have been introduced and investigated, such as rainbow, monochromatic, proper connection and conflict-free connection, both in edge-colored and vertex-colored graphs. We refer to, e.g.,  \cite{CzapJV18,LiM15,LiSS13,LiS17,LiW18,LiZZMZJ20} 
for more details. 
Besides of theoretical motivations, the colored constraints make the new versions of graph connectivity applicable to practical problems such as security and accessibility in communication networks (\cite{ChartrandJMZ08}); see also~\cite{BrauseJS18} for a meta-concept of graph connectivity with practical application scenarios. 

Motivated by conflict-free graph and hypergraph colorings which are useful models to solve, e.g., the problem of assigning frequencies to different base stations in cellular networks~\cite{EvenLRS03,Smorodinsky13} (see also~\cite{ChangJLZ21}),  the concepts of conflict-free (edge-)connection and conflict-free vertex-connection have been introduced in~\cite{CzapJV18} and~\cite{LiZZMZJ20}, respectively. Since then, many research papers have focused on this topic; see, e.g.,~\cite{ChangJLZ21,DoanHS22,DoanS21,JiLZ20,LiZ20,LiW18+,LiW23} and the recent survey~\cite{ChangH24}.  

An edge-colored (vertex-colored) graph is \emph{conflict-free (vertex-)connected} if any two distinct vertices are connected by a \emph{conflict-free path}, a path that is such that there is some color that occurs on exactly one edge (vertex) on the path. The \emph{conflict-free connection (vertex-connection) number} of a graph~$G$, written $\mathrm{cfc}(G)$ ($\mathrm{vcfc}(G)$), is the smallest color number that makes $G$ conflict-free (vertex-)connec\-ted. It turned out that, while computing these parameters is easy for graphs with at most one cut-vertex, the authors of \cite{CzapJV18,LiZZMZJ20} noted that it is very hard to determine $\mathrm{cfc}(T)$ and $\mathrm{vcfc}(T)$ for trees $T$; 
for paths $P$, $\mathrm{cfc}(P)$ and $\mathrm{vcfc}(P)$ are known.  
Also, the \emph{strong} version of conflict-free connection has been investigated~\cite{JiL20,JiLZ20}: an edge-colored graph is said to be \emph{strongly conflict-free connected} if any two distinct vertices are connected by a conflict-free \emph{shortest} path.  
Among others, strongly conflict-free connection $2$-colorable cubic graphs have been characterized in~\cite{JiL20}. It is shown in~\cite{JiLZ20} that computing the strong conflict-free connection number is \NP-hard. 
The complexity of computing the conflict-free (vertex-)connection number is still open. 

In this paper, we address the \emph{strong} version of conflict-free vertex-connection: we require that any two distinct vertices in a vertex-colored graph are connected by a conflict-free \emph{shortest} path. To the best of our knowledge, this strong vertex-version has not been considered before. 
%
\begin{definition}[Strong conflict-free vertex-connection coloring]
Let $G$ be a graph and $k\ge 1$ be an integer. A function $f:V(G)\to [k]=\{1,2,\ldots,k\}$  is a \emph{strong conflict-free vertex-connection $k$-coloring} of~$G$ if any two distinct vertices~$u$ and~$v$ of $G$ are connected by a \emph{strong conflict-free $u,v$-path}, a \emph{shortest} $u,v$-path that is such that there is a color $c\in [k]$ that occurs on exactly one vertex of the path. 

\noindent
The \emph{strong conflict-free vertex-connection number} of $G$, denoted $\svcfc(G)$, is the smallest integer $k$ such that $G$ has a strong conflict-free vertex-connection $k$-coloring.
\end{definition}
It is interesting to note that, in contrast to conflict-free vertex-connection colorings, strong conflict-free vertex-connection colorings are proper vertex-colorings; a \emph{proper vertex-coloring}, or just a \emph{coloring}, of a graph is an assignment of colors in $[k]$ to the vertices with no monochromatic edges. Graph coloring is another fundamental topic in graph theory and combinatorial optimization, with wide-ranging applications and many interesting open questions. 
For this reasons, many variants of graph coloring have been introduced and intensively studied. 
To the best of our knowledge, the concept of strong conflict-free vertex-connection coloring has not been considered before as a variant of the ordinary graph coloring, and our study on strong conflict-free vertex-connection colorings is also motivated by this fact. 

This paper focuses on the computational complexity of deciding, for a given graph $G$ and an integer $k$, whether $G$ has a strong conflict-free vertex-connection $k$-coloring. Below, we formally describe the problems as new variants of the classical {\sc chromatic number} and $k$-{\sc coloring} problems. 

\medskip\noindent
\fbox{
\begin{minipage}{.96\textwidth}
\textsc{strong conflict-free vertex-connection number} ({\sc svcfc})\\[.7ex]
\begin{tabular}{l l}
{\em Instance:}& A connected graph $G$ and an integer $k$.\\
{\em Question:}& $\svcfc(G)\le k$, i.e., does $G$ have a strong conflict-free\\
                       & vertex-connection $k$-coloring\,?
\end{tabular}
\end{minipage}
}

\medskip\noindent
When $k$ is a given constant, i.e., $k$ is not part of the input, we write $k$-{\sc svcfc} instead of {\sc svcfc}:

\medskip\noindent
\fbox{
\begin{minipage}{.96\textwidth}
\textsc{$k$-svcfc}\\[.7ex]
\begin{tabular}{l l}
{\em Instance:}& A connected graph $G$.\\
{\em Question:}& $\svcfc(G)\le k$, i.e., does $G$ have a strong conflict-free\\
                       & vertex-connection $k$-coloring\,?
\end{tabular}
\end{minipage}
}

\medskip\noindent
Note that, unlike graph colorings, strong conflict-free vertex-connection colorings are \emph{not} monotone: deleting vertices may turn the instance from yes to no. This fact indicates that determining the complexity of {\sc svcfc} and $k$-{\sc svcfc} may be harder than that of {\sc chromatic number} and $k$-{\sc coloring}, respectively.

\paragraph{Our results.} The contributions of this paper may be described as follows. Our first basic result is that it can be verified in polynomial time whether a given coloring is a strong conflict-free vertex-connection coloring, that is, {\sc svcfc} and $k$-{\sc svcfc} are in \NP. Note that this task is not obvious because the number of shortest paths between two vertices may be exponential in the vertex number of the input graph.  We then prove the following hardness results, where item (2) is the most interesting one.
\begin{itemize}
\item[(1)] $k$-{\sc svcfc} is \NP-complete for diameter-$d$ graphs for all pairs $(k,d)$ with $k\ge 3$ and $d\ge 2$ except when $(k,d)=(3,2)$.
\end{itemize}
It turns out that the only open case $(k,d)=(3,2)$ 
is in fact the famous long-standing open problem of determining the computational complexity of the classical $3$-{\sc coloring} for diameter-$2$ graphs.
\begin{itemize} 
\item[(2)] $3$-{\sc svcfc} remains \NP-complete even when restricted to $3$-colorable graphs with diameter~$3$, radius~$2$ and domination number~$3$, and cannot be solved in $2^{o(n)}$ time on $3$-colorable $n$-vertex graphs with diameter~$3$, radius~$2$ and domination number~$3$, unless the Exponential Time Hypothesis fails.
\end{itemize}
Note that, in item (2), the restriction on graphs with domination number~$3$ is quite strong when compared to the $3$-{\sc coloring} problem. In fact, deciding whether a graph with bounded domination number is $3$-colorable can be done in polynomial time (\cite[Theorem 5]{JansenK13}). 
Note that also the ETH-lower bound in (2) is stronger than the known one for $3$-coloring graphs with diameter~$3$ and radius~$2$: assuming ETH, $3$-{\sc coloring} cannot be solved on $n$-vertex graphs with diameter~$3$ and radius~$2$ in $2^{o(\sqrt{n})}$ time (\cite{MertziosS16}). 
We also provide two polynomially solvable cases: 
\begin{itemize}
\item[(3)] {\sc svcfc} is solvable in polynomial time when restricted to split graphs and to co-bipartite graphs (complements of bipartite graphs). 
\end{itemize}
In fact, we point out that an \emph{optimal} strong conflict-free vertex-connection coloring, one with minimum color number, of a given split graph or a co-bipartite graph can be computed in linear time.

\noindent
\paragraph{Related work.} Since paths in trees are unique, the concepts of conflict-free vertex-connection colorings and strong conflict-free vertex-connection colorings in trees coincide. Thus, strong conflict-free vertex-connection colorings in trees have been investigated implicitly in~\cite{LiZZMZJ20,LiZ20}. In~\cite{LiZZMZJ20}, Li et al. realized that determining the strong conflict-free vertex-connection number of trees is \lq\lq very difficult.\rq\rq\ For $n$-vertex paths $P_n$, they proved that $\svcfc(P_n)=\lceil\log_2(n+1)\rceil$. Several upper bounds for the (strong) conflict-free vertex-connection number of trees are given in~\cite{LiZZMZJ20} and the exact value of $\svcfc(T)$ for trees $T$ of diameter at most~$4$ has been determined in~\cite{LiZ20}. It follows from~\cite{LiW18+} that $\svcfc(T)$ is upper bounded by $\lceil\log_2(n+1)\rceil$ for $n$-vertex trees~$T$.

\noindent
\paragraph{Organization.} We provide preliminaries with some basic facts on strong conflict-free vertex-connection colorings and introduce some notation in section~\ref{sec:pre}. In section~\ref{sec:NP} we give a polynomial-time algorithm that verifies whether a given coloring is a strong conflict-free vertex-connection coloring. In section~\ref{sec:k>=4} we address the complexity of $k$-{\sc svcfc} for $k\ge 4$,  and in section~\ref{sec:k=3} the complexity of $3$-{\sc svcfc}. In section~\ref{sec:P} we describe polynomial-time algorithms for computing an optimal strong conflict-free vertex-connection coloring in split graphs and in co-bipartite graphs. Section~\ref{sec:con} concludes the paper with some open questions for further research. 


\section{Preliminaries}\label{sec:pre}
We consider only finite, simple and \emph{connected} undirected graphs $G=(V,E)$ with vertex set $V(G)=V$ and edge set $E(G)=E$. 
A \emph{clique} (an \emph{independent set}) is a set of pairwise (non-)adjacent vertices. As usual, $\chi(G)$ and $\omega(G)$ denote the chromatic number and the clique number of $G$, respectively. The following is straightforward.
\begin{proposition}\label{pro:strong->proper} 
Every strong conflict-free vertex-connection coloring is a proper vertex-coloring. In particular, for any graph $G$, $\omega(G)\le \chi(G)\le \svcfc(G)\le |V(G)|$. 
\end{proposition}
Notice that the gap between $\svcfc(G)$ and $\chi(G)$ can be arbitrary large: while $\chi(P_n)=2$ for all $n$-vertex path $P_n$ with $n\ge 2$, it has been proved in~(\cite[Theorem 2.1]{LiZZMZJ20}) that $\svcfc(P_n)=\lceil \log_2(n+1)\rceil$. 

The \emph{distance} between two vertices $u$ and $v$, written $\dist(u,v)$, is the length of a shortest $u,v$-path. 
The \emph{diameter} of~$G$, written $\diam(G)$, is the maximum distance between any two vertices in~$G$, $\diam(G)=\max\{\dist(u,v)\mid u,v\in V(G)\}$, and $\rad(G)=\min_{u\in V(G)}\{\max\{\dist(u,v)\mid v\in V(G)\}\}$ is the \emph{radius} of~$G$. 
We use $N(v)$ to denote the neighborhood of a vertex $v$. Observe that a graph~$G$ has diameter at most~$2$ if and only if, for any two non-adjacent vertices $u$ and $v$ of~$G$, $N(u)\cap N(v)\not=\emptyset$. That is, a graph has diameter at most~$2$ if and only if, two non-adjacent vertices have a common neighbor.  The following is straightforward.
\begin{proposition}\label{pro:diam2}
In any diameter-$2$ graph, proper vertex-colorings and strong conf\-lict-free vertex-connection colorings coincide.
\end{proposition}
A \emph{cograph}, or a graph without an induced $4$-vertex path $P_4$, has diameter at most~$2$.  Proposition~\ref{pro:diam2} immediately implies that {\sc svcfc} is polynomially solvable on cographs\footnote{It is well known that the chromatic number of cographs can be computed in linear time}, and that the complexity of $3$-{\sc svcfc} restricted to graphs of diameter~2 is the same of $3$-{\sc coloring} in graphs of diameter~$2$. Determining the complexity of $3$-{\sc coloring} in diameter-$2$ graphs, hence the complexity of $3$-{\sc svcfc} in diameter-$2$ graphs, is a notoriously difficult, well-known long-standing open problem in algorithmic graph theory (cf.~\cite{MertziosS16,DebskiPR22}).

Properly $2$-colorable graphs are called \emph{bipartite}. Equivalently, a graph is bipartite if and only if its vertex set can be partitioned into two independent sets. A \emph{complete bipartite} graph is a bipartite graph $G=(V,E)$ with a bipartition $V=A\,\dot\cup\,B$ into independent sets $A$ and $B$ such that $E=\{uv\mid u\in A, v\in B\}$. It is easy to see that a graph is complete bipartite if and only if it is (connected) bipartite and contains no induced $4$-vertex path $P_4$. 

Graphs $G$ having a strong conflict-free vertex-connection $2$-coloring, that is, graphs $G$ with $\svcfc(G)\allowbreak\le2$, can be characterized as follows. 
\begin{proposition}
\label{pro:k=2}
A graph is strongly conflict-free vertex-connection $2$-colorable if and only if it is a complete bipartite graph.
\end{proposition}
\begin{proof}
A complete bipartite graph has diameter~2, hence its (unique) 2-coloring is also strong conflict-free vertex-connection $2$-coloring (by Proposition~\ref{pro:diam2}). 

Conversely, suppose that $G$ admits a strong conflict-free vertex-connection $2$-coloring~$f$. Then $G$ is connected and, since~$f$ is particularly a proper $2$-coloring of $G$ (by Proposition~\ref{pro:strong->proper}), $G$ is a bipartite graph. 
Moreover,~$G$ cannot contain any induced $P_4$ $u-v-w-x$, otherwise the distance between~$u$ and~$x$ is three (as~$G$ is bipartite) and every shortest $u,x$-path in~$G$ has no unique color (as $f$ uses only two colors). Thus, $G$ is a connected $P_4$-free bipartite graph, or equivalently,~$G$ is complete bipartite. 
\qed
\end{proof}
As a corollary from Proposition~\ref{pro:k=2}, $2$-{\sc svcfc} can be solved in linear time. 

A \emph{co-bipartite} graph is the complement of a bipartite graph. Equivalently, $G=(V,E)$ is co-bipartite if the vertex set can be partitioned into two cliques. 
A \emph{split graph} is one whose vertex set can be partitioned into a clique and an independent set. 
Observe that (connected) co-bipartite graphs and split graphs have diameter at most~$3$. 
It is well known that an optimal proper coloring of a given co-bipartite graph and a split graph can be computed in linear time. 

Algorithmic lower bounds in this paper are conditional, based on the Exponential Time Hypothesis (ETH for short)~\cite{ImpagliazzoP01}. The ETH asserts that no algorithm can solve \3SAT\ in subexponential time $2^{o(n)}$ for $n$-variable \textsc{3-cnf} formulas. As shown by the Sparsification Lemma in~\cite{ImpagliazzoPZ01}, the hard cases of \3SAT\ consist of sparse formulas with $m=O(n)$ clauses. Hence, the ETH implies that \3SAT\ cannot be solved in $2^{o(n+m)}$ time. 

It is known (see, e.g.,~\cite[Theorem 3.2]{LokshtanovMS11}) that, assuming ETH, $3$-{\sc coloring} cannot be solved in $2^{o(n)}$ time on $n$-vertex graphs. This fact can be immediately extended for $k$-{\sc coloring}, for any fixed integer $k\ge4$: given a graph $G$, let $G'$ be obtained from $G$ by adding a new vertex and joining it to all vertices in~$G$. Then $G$ is $(k-1)$-colorable if and only if $G'$ is $k$-colorable. Thus, assuming ETH, $k$-{\sc coloring} cannot be solved in $2^{o(n)}$ time on $n$-vertex graphs for any fixed $k\ge 3$.

\section{Verifying strong conflict-free vertex-connection colorings}\label{sec:NP} 
In this section we show that, given a graph $G$ together with a $k$-coloring of $G$, it can be verified  in polynomial time whether the $k$-coloring is a strong conflict-free vertex-connection $k$-coloring of $G$. 
We remind that it is not obvious how this can be done in polynomial time since the number of shortest paths between two vertices may be exponential in the vertex number of the input graph. 
Our algorithm has some similarity to the checking if a given edge-coloring of a graph is a strong conflict-free connection edge-coloring proposed in~\cite[Algorithm 2]{JiLZ20}.

Let $f:V(G)\to [k]$ be a coloring of an $n$-vertex graph $G$. 
Note that we may assume that $k\le n=|V(G)|$.  We have to check whether any two distinct vertices are connected by a strong conflict-free path under the coloring~$f$. 
For two vertices $u$ and $v$ we first build a \emph{level graph} $L[u,v]$ that contains all shortest $u,v$-paths in~$G$ as follows. 
\begin{itemize}
\item Begin with breadth-first search (BFS), starting at $u$, to compute the distance $d=\dist_G(u,v)$, and the distance levels $L_i$, $0\le i\le d-1$, 
\[L_i=\{x\in V(G)\mid \dist_G(u,x)=i\}.\]
\item Let $L_d=\{v\}$ and note that $L_0=\{u\}$.
\item Then $L[u,v]$ is obtained from all levels $L_i$, $0\le i\le d$, and all \emph{vertical edges} of $G$, that is, edges between levels $L_{i-1}$ and $L_i$, $1\le i\le d$. 
\end{itemize}
Note that any shortest $u,v$-path in $L[u,v]$ is a shortest $u,v$-path in $G$ and  any shortest $u,v$-path in $G$ is a shortest $u,v$-path in $L[u,v]$. So to check, if there is a strong conflict-free $u,v$-path in $G$, it suffices to check if one exists in the level graph $L[u,v]$. Note also that, as $f$ is a (proper) coloring, we have only to consider any two vertices at distance at least~$3$. 

Let $V_c$ be the set of vertices in $L[u,v]$ colored by color~$c\in [k]$ under coloring~$f$,
\[V_c=\{x\mid x\in L_0\cup L_1\cup\cdots\cup L_d\,\text{ and}\, f(x)=c\}.\] 
Then, to check whether there is a strong conflict-free $u,v$-path in $L[u,v]$, we have to check whether there is a shortest $u,v$-path on which a color $c$ occurs uniquely (1) on $u$ or on $v$, or (2) on a vertex in level $L_i$ for some $1\le i\le d-1$. Clearly, the first case is equivalent to the facts that $f(u)\not=f(v)$ and $L[u,v]-V_c$ is connected, and the second case is equivalent to the connectedness of $L[u,v]-(V_c\cup (L_i\setminus\{x\}))$ for some $x\in V_c\cap L_i$.   
The details are given in Algorithm~\ref{algo:checkingstrong}. 

\begin{algorithm}[!ht]
\SetKwInOut{Input}{input}\SetKwInOut{Output}{output}
\DontPrintSemicolon 
\Input{A connected graph $G=(V,E)$, a proper coloring $f:V\to [k]$ and two vertices $u, v$ of $G$.}
\Output{\lq\lq strong\rq\rq\ if $u$ and $v$ are connected by a conflict-free shortest $u,v$-path under $f$ and \lq\lq not strong\rq\rq\ otherwise.}

\BlankLine
compute the level graph $L[u,v]$ using BFS\;
\If{$f(u)\not=f(v)$}{
 \ForEach{color $c \in \{f(u),f(v)\}$}{ 
         \If{$L[u,v]-V_c$ is connected}{
         \KwRet \lq\lq strong\rq\rq}
         }
}
\ForEach{color $c\in [k]\setminus\{f(u),f(v)\}$}{
   \For{$i\leftarrow 1$ \KwTo $d-1$}{
      \ForEach{vertex $x\in V_c\cap L_i$}{ 
         \If{$L[u,v]-(V_c\cup (L_i\setminus\{x\}))$ is connected}{
         \KwRet \lq\lq strong\rq\rq}
      }
  }
}
\KwRet \lq\lq not strong\rq\rq

\caption{verifying strong conflict-free vertex-connection coloring}
\label{algo:checkingstrong}
\end{algorithm}

By the above, the correctness of Algorithm~\ref{algo:checkingstrong} is immediate.

Since computing the level graph and testing connectivity based on BFS requires $O(n+m)$ time, the runtime of Algorithm~\ref{algo:checkingstrong} takes $O(n+m)+k\cdot d\cdot n\cdot O(n+m)=O(kn^2)$ steps. 

To verify whether a given $k$-coloring is a strong conflict-free vertex-connection coloring, we call Algorithm~\ref{algo:checkingstrong} for at most $n^2$ pairs of vertices $u$ and $v$. Thus the total runtime needed  is at most $O(kn^4)$, and we obtain:
\begin{theorem}\label{thm:checkingstrong}
There is a polynomial-time algorithm that correctly verifies whether a proper vertex-coloring of a connected graph is a strong conflict-free vertex-connection coloring.
\end{theorem} 

\begin{corollary}
{\sc svcfc}, and hence $k$-{\sc svcfc}, is in \NP. 
\end{corollary}

\section{Hardness result: $k\ge 4$ colors}\label{sec:k>=4}
In this section we will show that, for any constant $k\ge 4$, $k$-{\sc svcfc} is \NP-complete, even when restricted to graphs of diameter exactly~$d$ for any given integer $d\ge 2$. 

The case of $d=2$ can be quickly seen by a simple reduction from $(k-1)$-{\sc coloring}: given a (non-complete) graph $G$, let $G'$ be obtained from~$G$ by taking a new vertex $u$ and adding an edge between~$u$ and each vertex of $G$. Then~$G'$ has diameter~2 and has a (strong conflict-free vertex-connection) $k$-coloring if and only $G$ has a $(k-1)$-coloring. Moreover, recall that, assuming ETH, $(k-1)$-{\sc coloring} cannot be solved in $2^{o(|V(G)|)}$ time, hence $k$-{\sc svcfc} cannot be solved in $2^{o(n)}$ time on diameter-$2$ $n$-vertex graphs under ETH.

Let $d\ge 3$. We first describe suitable graphs of any given diameter~$d$ that are strongly conflict-free vertex-connection $3$-colorable. 

For any integer $n\ge 1$, let $Q_n$ be the graph with~$3n+1$ vertices $a_i, b_i$, $1\le i\le n$, and $c_i$, $0\le i\le n$, and~$5n$ edges $c_ia_{i+1}$, $0\le i\le n-1$, $c_ia_i$, $1\le i\le n$, $c_ib_{i+1}$, $0\le i\le n-1$, $c_ib_i$, $1\le i\le n$, and $a_ib_i$, $1\le i\le n$. 
For any integer $n\ge 2$, let $R_n$ be the graph obtained from $Q_{n-1}$ by taking a new vertex $c_n$ and adding the edge between $c_n$ and the vertex $c_{n-1}$ in $Q_{n-1}$. 
See Fig.~\ref{fig:QR} for the graph $Q_4$ and $R_4$. 

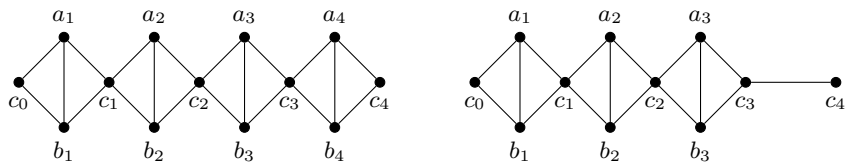
\begin{figure}
\centering
\begin{tikzpicture}[scale=.3]
\node[vertex] (c0) at (1,3) [label=below:{\small $c_0$}] {};
\node[vertex] (a1) at (3,5) [label=above:{\small $a_1$}] {};
\node[vertex] (b1) at (3,1) [label=below:{\small $b_1$}] {};
\node[vertex] (c1) at (5,3) [label=below:{\small $c_1$}] {};
\node[vertex] (a2) at (7,5) [label=above:{\small $a_2$}] {};
\node[vertex] (b2) at (7,1) [label=below:{\small $b_2$}] {};
\node[vertex] (c2) at (9,3) [label=below:{\small $c_2$}] {};
\node[vertex] (a3) at (11,5) [label=above:{\small $a_3$}] {};
\node[vertex] (b3) at (11,1) [label=below:{\small $b_3$}] {};
\node[vertex] (c3) at (13,3) [label=below:{\small $c_3$}] {};
\node[vertex] (a4) at (15,5) [label=above:{\small $a_4$}] {};
\node[vertex] (b4) at (15,1) [label=below:{\small $b_4$}] {};
\node[vertex] (c4) at (17,3) [label=below:{\small $c_4$}] {};

\draw (c0)--(a1)--(c1)--(a2)--(c2)--(a3)--(c3)--(a4)--(c4);
\draw (c0)--(b1)--(c1)--(b2)--(c2)--(b3)--(c3)--(b4)--(c4);
\draw (a1)--(b1); \draw (a2)--(b2); \draw (a3)--(b3); \draw (a4)--(b4);
\end{tikzpicture} 
\qquad
\begin{tikzpicture}[scale=.3]
\node[vertex] (c0) at (1,3) [label=below:{\small $c_0$}] {};
\node[vertex] (a1) at (3,5) [label=above:{\small $a_1$}] {};
\node[vertex] (b1) at (3,1) [label=below:{\small $b_1$}] {};
\node[vertex] (c1) at (5,3) [label=below:{\small $c_1$}] {};
\node[vertex] (a2) at (7,5) [label=above:{\small $a_2$}] {};
\node[vertex] (b2) at (7,1) [label=below:{\small $b_2$}] {};
\node[vertex] (c2) at (9,3) [label=below:{\small $c_2$}] {};
\node[vertex] (a3) at (11,5) [label=above:{\small $a_3$}] {};
\node[vertex] (b3) at (11,1) [label=below:{\small $b_3$}] {};
\node[vertex] (c3) at (13,3) [label=below:{\small $c_3$}] {};
\node[vertex] (c4) at (17,3) [label=below:{\small $c_4$}] {};

\draw (c0)--(a1)--(c1)--(a2)--(c2)--(a3)--(c3)--(c4);
\draw (c0)--(b1)--(c1)--(b2)--(c2)--(b3)--(c3);
\draw (a1)--(b1); \draw (a2)--(b2); \draw (a3)--(b3); 
\end{tikzpicture} 
\caption{The graph $Q_4$ (left) and $R_4$ (right).
}\label{fig:QR}
\end{figure}

\begin{observation}
\label{obs:2}
$Q_n$ has diameter $2n$, $R_n$ has diameter $2n-1$. Both $Q_n$ and $R_n$ are strongly conflict-free vertex-connection $3$-colorable.
\end{observation}
\begin{proof}
The diameter of $Q_n$, respectively, $R_n$, is the distance between $c_0$ and~$c_n$, which is $2n$ in $Q_n$ and $2n-1$ in $R_n$. Color all vertices $a_i$ with color~$a$, all vertices~$b_i$ with color~$b$. 
Then, in $Q_n$, color all vertices~$c_i$ with color~$c$. In $R_n$, color all vertices~$c_i$ but $c_n$ with color~$c$ and vertex $c_n$ with color~$a$. Then, it can be immediately verified that the obtained coloring is a strong conflict-free connection $3$-coloring of $Q_n$, respectively, $R_n$. \qed 
\end{proof}

\begin{theorem}
\label{thm:k-and-diam}
For any constants $d\ge 2$ and $k\ge 4$, $k$-{\sc svcfc} is \NP-complete, even when restricted on diameter-$d$ graphs, and assuming ETH, cannot be solved in $2^{o(n)}$ time on diameter-$d$ $n$-vertex graphs. 
\end{theorem}
\begin{proof}
We have seen the proof in case $d=2$ at the beginning of this section. Let $d\ge 3$. We give a polynomial-time reduction from $(k-1)$-{\sc coloring}: Given a graph $G$, construct a graph $G'$ from $G$ as follows.
\begin{itemize}
\item If $d\ge 3$ is odd, take $Q_n$ with $n=(d-1)/2$, and make the vertex $c_0$ in $Q_n$ adjacent to all vertices of $G$.
\item If $d\ge 4$ is even, take $R_n$ with $n=d/2$, and make the vertex $c_0$ in $R_n$ adjacent to all vertices of $G$.
\end{itemize}
\noindent
Then $\diam(G')=d$ (by construction and by Observation~\ref{obs:2}), and $G$ is $(k-1)$-colorable if and only if $G'$ is strongly conflict-free vertex-connection $k$-colorable. 
Assume that $G$ has a $(k-1)$-coloring. Let us consider the case $d$ is odd (the case~$d$ is even is similar). Then color~$a_i$ with color~$1$, $b_i$ with color~$2$, $c_i$ with color~$3$, $1\le i\le n-1$, and color $c_0$ with a new color~$k$. This yields a strong conflict-free vertex-connection $k$-coloring of~$G'$: by Observation~\ref{obs:2} and by the fact that color~$k$ is used only on $c_0$, it remains to consider any two non-adjacent vertices in~$G$. A shortest path in~$G'$ connecting such two vertices containing $c_0$ has the unique color~$k$.

The other direction is clear: if $G'$ admits a (strong conflict-free vertex-connec\-tion) $k$-coloring, then the color of the vertex $c_0$ does not occur in $G$, hence $G$ is $(k-1)$-colorable.

For the second part, note that $G'$ has $N=|V(G)|+3\cdot\frac{d-1}{2}+1=O(|V(G)|)$ vertices (if $d$ is odd) and $N=|V(G)|+3\cdot\frac{d}{2}-1=O(|V(G)|)$ (if $d$ is even). Recall that, assuming ETH, $(k-1)$-{\sc coloring} cannot be solved in $2^{o(|V(G)|)}$ time. Hence we can conclude that $k$-{\sc svcfc} cannot be solved in $2^{o(N)}$ time on diameter-$d$ $N$-vertex graphs, unless ETH fails.
\qed 
\end{proof}

Note that the reduction in this section does not work for $k=3$.  The complexity of $3$-{\sc svcfc} will be handled in the next section.

\section{Hardness result: $k=3$ colors}\label{sec:k=3}
A vertex set $D\subset V(G)$ of graph $G$ is a \emph{dominating set} if every vertex in $V(G)\setminus D$ has a neighbor in $D$. The \emph{domination number} of $G$ is the smallest size of a dominating set in $G$. It is well-known that $3$-{\sc coloring} is polynomially solvable when restricted to graphs of bounded domination number (see~\cite[Theorem 5]{JansenK13}), and assuming ETH, cannot be solved in $2^{o(\sqrt{n})}$ time on $n$-vertex graphs with diameter~$3$ and radius~$2$ (see~\cite{MertziosS16}). 

In this section, we prove our main hardness result for $3$-{\sc svcfc}, Theorem~\ref{thm:3333} below, which implies in a sense that $3$-{\sc svcfc} is much harder than the classical $3$-{\sc coloring} problem.

\begin{theorem}
\label{thm:3333}
$3$-{\sc svcfc} is \NP-complete, even when restricted to $3$-colorable graphs with diameter~$3$, radius~$2$ and domination number~$3$. 
Moreover, assuming ETH, no algorithm with runtime $2^{o(n)}$ can solve $3$-{\sc svcfc} on $n$-vertex  graphs with diameter~$3$, radius~$2$, domination number~$3$ and with a given $3$-coloring.
\end{theorem}
\begin{proof}
We give a polynomial-time reduction from $3$-{\sc sat} 
to $3$-{\sc svcfc}. Let $\phi$ be a {\sc 3-cnf}-formula with $m$ clauses and $n$ variables. 
We construct a graph $G$ as follows (see also Fig.\ref{fig:3333}).

\begin{itemize}
\item
We begin with four \emph{special vertices}, $a$, $b,c$ and~$d$, with edges $ab$, $ac$, $ad$, $bd$ and $cd$. 
\item 
For each variable $x$, create two adjacent \emph{literal vertices} $v_x$ and $v_{\overline{x}}$, and make~$v_x$ and~$v_{\overline{x}}$ adjacent to the special vertex~$a$. 
\item
For each clause $C_j$, create a triangle $c_j, c_j^a,c_j^b$. Call the vertex~$c_j$ a \emph{clause vertex}, and make it adjacent to those literal vertices if the corresponding literal is contained in clause~$C_j$.
\item 
For each $j\in [m]$, 
\begin{itemize}
   \item make $c_j$ adjacent to the special vertex $c$, 
   \item make $c_j^a$ adjacent to the special vertex $a$, and 
   \item make $c_j^b$ adjacent to the special vertex $b$.
\end{itemize}
\end{itemize}

\begin{figure}
\centering
\begin{tikzpicture}[scale=.3]
\node[vertex] (c1) at (.5,9.5) [label=left:{\small $c_1$}] {};
\node[vertex] (c1a) at (4,9) [label=right:{\small $c_1^a$}] {};
\node[vertex] (c1b) at (.5,6) [label=left:{\small $c_1^b$}] {};
\node[vertex] (c2) at (8.5,10) [label=left:{\small $c_2$}] {};
\node[vertex] (c2a) at (11.5,9) [label=right:{\small $c_2^a$}] {};
\node[vertex] (c2b) at (9,6) [label=left:{\small $c_2^b$}] {};
\node[vertex] (c3) at (23.5,9.5) [label=right:{\small $c_3$}] {};
\node[vertex] (c3a) at (20.5,9) [label=left:{\small $c_3^a$}] {};
\node[vertex] (c3b) at (24,6) [label=right:{\small $c_3^b$}] {};
\node[vertex] (c4) at (31.5,10) [label=right:{\small $c_4$}] {};
\node[vertex] (c4a) at (28,7) [label=below:{\small $c_4^a$}] {};
\node[vertex] (c4b) at (31.5,6) [label=right:{\small $c_4^b$}] {};
\node[vertex] (x) at (4,17) [label=above:{\small $v_x$}] {};
\node[vertex] (x-) at (8.5,17) [label=above:{\small $v_{\overline{x}}$}] {};
\node[vertex] (y) at (13,17) [label=above:{\small $v_y\,$}] {};
\node[vertex] (y-) at (19,17) [label=above:{\small $v_{\overline{y}}$}] {};
\node[vertex] (z) at (23,17) [label=above:{\small $v_z$}] {};
\node[vertex] (z-) at (27.5,17) [label=above:{\small $v_{\overline{z}}$}] {};
\node[vertex] (a) at (16,22) [label=above:{\small $a$}] {};
\node[vertex] (b) at (19.5,2.5) [label=below:{\small $b$}] {};
\node[vertex] (c) at (12.5,2.5)  [label=below:{\small $c$}] {};
\node[vertex] (d) at (16,8) [label=right:{\small $d$}] {};

\draw[thick] (c1)--(c1a)--(c1b)--(c1); \draw[thick] (c2)--(c2a)--(c2b)--(c2); \draw[thick] (c3)--(c3a)--(c3b)--(c3); \draw[thick] (c4)--(c4a)--(c4b)--(c4);
\draw[thick] (x)--(x-); \draw[thick] (y)--(y-); \draw[thick] (z)--(z-);
\draw[gray] (x)--(a)--(x-); \draw[gray] (y)--(a)--(y-); \draw[gray] (z)--(a)--(z-);
\draw[gray] (a)[bend angle=5,bend right] to(c1a); 
\draw[gray] (a)--(c2a); \draw[gray] (a)--(c3a); 
\draw[gray] (a)[bend angle=7,bend left]to (c4a);
\draw[gray] (a)--(b); \draw[gray] (a)--(c); \draw[gray] (a)--(d); \draw[gray] (b)--(d)--(c);
\draw[gray] (b)--(c1b); \draw[gray] (b)--(c2b); \draw[gray] (b)--(c3b); \draw[gray] (b)[bend angle=3, bend right]to (c4b);
\draw[gray] (c)[bend angle=6, bend left]to (c1); \draw[gray] (c)--(c2); \draw[gray] (c)--(c3); \draw[gray] (c)--(c4);

\draw (x-)--(c1)--(y); \draw (z)--(c1);
\draw (x)--(c2)--(y-); \draw (z)--(c2);
\draw (x)--(c3)--(y); \draw (y)--(c3); \draw (z-)--(c3);
\draw (x-)--(c4)--(y-); \draw (z-)--(c4);
\end{tikzpicture} 
\caption{The graph $G$ obtained from the formula $\phi$ with $C_1=\{\overline{x},y,z\}$, $C_2=\{x,\overline{y},z\}$, $C_3=\{x,y,\overline{z}\}$ and $C_4=\{\overline{x},\overline{y},\overline{z}\}$.
}\label{fig:3333}
\end{figure}
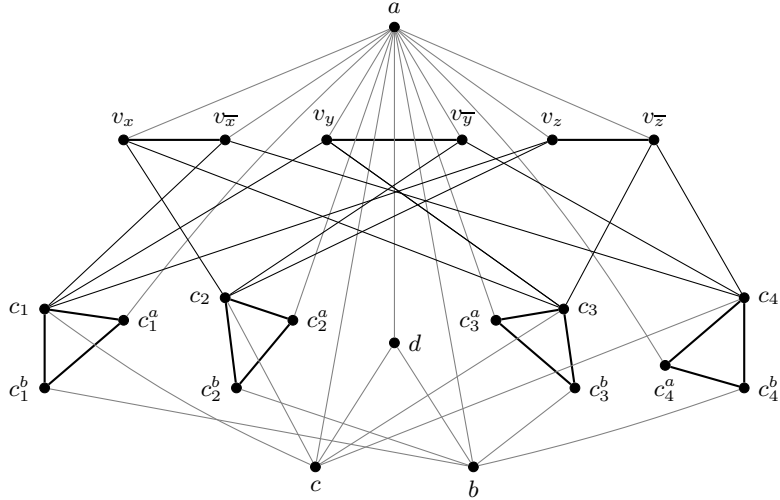

Observe that $G$ has $3m+2n+4$ vertices and can be created in polynomial time. 

\begin{lemma}\label{lem:333}
$G$ is $3$-colorable, has domination number~$3$, diameter~$3$ and radius~$2$.
\end{lemma}
\emph{Proof of Lemma~\ref{lem:333}.} 
Observe that the vertex set of~$G$ can be partitioned into three independent sets, $V(G)=I_1\,\dot\cup\, I_2\,\dot\cup\, I_3$, with 
\begin{itemize}
\item $I_0=\{v_{x}\mid \text{$x$ is a variable}\}\cup \{c_j^a\mid j\in [m]\}\cup\{b,c\}$,
\item $I_1=\{v_{\overline{x}}\mid \text{$x$ is a variable}\}\cup \{c_j^b\mid j\in [m]\}\cup \{d\}$, and
\item $I_2=\{c_j\mid j\in [m]\}\cup\{a\}$.
\end{itemize}
That is, $G$ is $3$-colorable. 
Observe next that 
\begin{itemize}
\item $N(a)=(I_0\cup I_1)\setminus\{c_j^b\mid j\in [m]\}$,
\item $N(b)=\{c_j^b\mid j\in [m]\}\cup \{a,d\}$, and
\item $N(c)=I_2\cup\{d\}$.
\end{itemize}
That is, $\{a,b,c\}$ is a dominating set of~$G$. It can be seen, by inspection, that there is no smaller dominating set.  
We now argue that $G$ has radius~$2$ and diameter~$3$. For $u\in V(G)$, let $\ecc(u)=\allowbreak\max\{\dist(u,z)\mid z\in V(G)\}$, the eccentricity of~$u$. Then $\rad(G)$ and $\diam(G)$ are the minimum and maximum eccentricity, respectively. By the previous observation and the fact that $d$ is adjacent to $a,b$ and $c$, we have $\ecc(a)=\ecc(b)=\allowbreak\ecc(c)=\ecc(d)=2$. Since $G$ has no universal vertex, it follows that $\rad(G)=2$. 
To see that $G$ has diameter~$3$, consider two non-adjacent vertices $u\not= w$ with no common neighbor. We argue that $\dist(u,w)= 3$.  
By the previous observation, $u,w\notin\{a,b,c,d\}$. 
If $u\in\{v_x,v_{\overline{x}}\}$ for some variable $x$, say $u=v_x$, then, as $u$ and $w$ have no common neighbor, $w\in\{c_j,c_j^b\}$ for some clause $C_j$ not containing~$x$, and $u-a-c_j^a-w$ is a $u,w$-path of length~$3$. The case $w\in\{v_x,v_ {\overline{x}}\}$ is similar. 
So, let $u\in\{c_j,c_j^a,c_j^b\}$ and $w\in\{c_{j'},c_{j'}^a,c_{j'}^b\}$ for some $j\not=j'$. If $u=c_j$ then, as $u$ and $w$ have no common neighbor, $w\not=c_{j'}$ and $u-c-c_{j'}-w$ is a $u,w$-path of length~$3$. If $u=c_j^a$ then $w\not=c_{j'}^a$ and $u-a-c_{j'}^a-w$ is a $u,w$-path of length~$3$. Similarly, if $u=v_j^b$ then $w\not=c_{j'}^b$ and $u-b-c_{j'}^b-w$ is a $u,w$-path of length~$3$. 
Thus, the diameter of $G$ is exactly~$3$, and the proof of Lemma~\ref{lem:333} is complete.

\begin{lemma}\label{lem:phi->G}
If $\phi$ is satisfiable then $G$ admits a strong conflict-free vertex-connection $3$-coloring.
\end{lemma}
\emph{Proof of Lemma~\ref{lem:phi->G}.} 
Given a satisfying assignment for $\phi$, define a $3$-coloring $f$ of~$G$ as follows. 
\begin{itemize}
\item For each variable $x$, 
  \begin{itemize}
      \item if $x=\text{True}$ then $f(v_x)=1$ and $f(v_{\overline{x}})=0$,
      \item if $x=\text{False}$ then $f(v_x)=0$ and $f(v_{\overline{x}})=1$,
  \end{itemize}
  \item $f(a)=f(c_j)=2$, $j\in [m]$,
  \item $f(b)=f(c)=0$, $f(d)=1$,
  \item $f(c_j^a)=0$, and $f(c_j^b)=1$, $j\in [m]$.
\end{itemize}
Clearly, $f$ is a proper $3$-coloring of $G$. We argue that $f$ is a strong conflict-free vertex-connection $3$-coloring. Since $f$ is a proper coloring of $G$, every pair of two distinct vertices at distance at most~2 are connected with a shortest conflict-free path. Since the diameter of~$G$ is three, it therefore remains to consider vertices at distance~$3$ in~$G$. 

Let $u$ and $w$ be two vertices at distance~3 in $G$. Then (see also the proof of Lemma~\ref{lem:333}) we have the following cases up to symmetry: 
\begin{itemize}
\item[(1)] 
$u=v_x$ and $w\in\{c_j, c_j^b\}$ for some variable $x$ and some $j$ with $x\not\in C_j$;
\item[(2)]  
$u=c_j$ and $w\in\{c_{j'}^a,c_{j'}^b\}$ for some $j\not=j'$;
\item[(3)]
$u=c_j^a$ and $w=c_{j'}^b$ for some $j\not=j'$.
\end{itemize}
We now inspect all these cases.

\smallskip\noindent
(1) Consider first $u=v_x$ and $w=c_j$ for some variable $x$ and some $j$ with $x\not\in C_j$. Let $x\in C_{j'}$. If $x=\text{True}$ then $v_x, c_{j'},c,c_j$ is a shortest $u,w$-path with the unique color $f(v_x)=1$. So, let $x=\text{False}$, and let $y$ be a variable such that the literal $\ell\in \{y,\overline{y}\}$ is a true literal in $C_{j}$. Note that $f(v_x)=0$. Then $v_x,a,v_{\ell},c_j$ is a shortest $u,w$-path with the unique color $f(v_{\ell})=1$. The case $u=v_{\overline{x}}$ and $w=c_j$ is similar. 

Consider next $u=v_x$ and $w=c_j^b$ for some variable $x$ and some $j$ with $x\not\in C_j$. The shortest $u,w$-path $v_x,a,c_j^a, c_j^b$ has the unique color $f(c_j^a)=0$ (if $x=\text{True}$) or the unique color $f(c_j^b)=1$ (otherwise). 


\smallskip\noindent
(2) Consider first $u=c_j$ and $w=c_{j'}^a$ for some $j\not=j'$. Let $y$ be a variable such that the literal $\ell\in \{y,\overline{y}\}$ is a true literal in $C_{j}$. Note that $f(v_{\ell})=1$. Then $c_j,v_{\ell}, a, c_{j'}^a$ is a shortest $u,w$-path with the unique colors $f(v_{\ell})=1$ and $f(c_{j'}^a)=0$.  

Consider next $u=c_j$ and $w=c_{j'}^b$ for some $j\not=j'$. The shortest $u,w$-path $c_j, c_j^b, b, c_{j'}^b$ has the unique colors $f(c_j)=2$ and $f(b)=0$.  

\smallskip\noindent
(3) $u=c_j^a$ and $w=c_{j'}^b$ for some $j\not=j'$: The shortest $u,w$-path $c_j^a, a, c_{j'}^a, c_{j'}^b$ has the unique colors $f(a)=2$ and $f(c_{j'}^b)=1$.  

\smallskip
We have seen that any two distinct vertices are connected by a shortest path with a unique color. This completes the proof of Lemma~\ref{lem:phi->G}.

\begin{lemma}\label{lem:G->phi}
If $G$ admits a strong conflict-free vertex-connection $3$-coloring then $\phi$ is satisfiable.
\end{lemma}
\emph{Proof of Lemma~\ref{lem:G->phi}.} 
Let $f:V(G)\to\{0,1,2\}$ be a strong conflict-free vertex-connection $3$-coloring of $G$. W.l.o.g. let $f(a)=2$. Then, for each variable $x$, the literal vertices $v_x$ and $v_{\overline{x}}$, as well as the special vertices $b,c$ and~$d$, are colored by color~$0$ or~$1$ as they are adjacent to the special vertex~$a$. Moreover, $f(v_x)\not=f(v_{\overline{x}})$ as $v_x$ and $v_{\overline{x}}$ are adjacent.  Let $f(d)=1$, say. Then $f(b)=f(c)=0$ as $b$ and $c$ are adjacent to $d$, and 
\begin{claim} 
$f(c_j)=2$ for all $j\in [m]$.  
\end{claim}
\emph{Proof of the Claim.} Suppose for a contradiction that, for some $j$, $f(c_j)\not=2$. Then $f(c_j)=1$ (as~$c_j$ is adjacent to~$c$ and $f(c)=0$) and $f(c_j^a)=0$ (as $c_j^a$ is adjacent to $a$ and $f(a)=2$), and therefore $f(c_j^b)=2$. 

Now let $x$ be a variable such that $x$ and $\overline{x}$ both are not contained in the clause~$C_j$. 
Note that the literal vertices $v_x$ and $v_{\overline{x}}$ are non-adjacent to the clause vertex~$c_j$, and that $v_x$ and $v_{\overline{x}}$ have no common neighbors with $c_j^b$. 
Thus, $c_j^b, c_j^a, a, v$ and $c_j^b, b, a, v$ are the shortest paths between~$c_j^b$ and $v\in\{v_x,v_{\overline{x}}\}$. Now, for $v\in\{v_x,v_{\overline{x}}\}$ with $f(v)=0$, none of the above shortest $c_j^b,v$-path has a unique color, a contradiction. 
Thus, $f(c_j)=2$ for all $j\in [m]$ as claimed.

It follows that $f(c_j^a)=0$ and $f(c_j^b)=1$ for all $j\in [m]$.

We now argue that $\phi$ is satisfied by assigning literals $\ell$ with True if $f(v_\ell)=1$ and False if $f(v_\ell)=0$:
suppose there is a clause $C_j$ with $f(v_{\ell})=0$ for all literals $\ell\in C_j$. 
Note that, for all $j'\not=j$, all shortest $c_j,c_{j'}^a$-paths are
\[c_j, v_\ell, a, c_{j'}^a \text { for all $\ell\in C_j$, and } c_j, v_{\ell}, c_{j'}, c_{j'}^a \text{ for all $\ell\in C_j\cap C_{j'}$},\] 
and
\[c_j,c,c_{j'},c_{j'}^a.\]
But none of them has a unique color, a contradiction. 

Thus, every clause $C_j$ must contain a literal $\ell$ such that $f(v_\ell)=1$, i.e., $\ell$ is a true literal in $C_j$. The proof of Lemma~\ref{lem:G->phi} is complete.

Now, the first part of Theorem~\ref{thm:3333} follows from Lemmata~\ref{lem:333}, \ref{lem:phi->G} and~\ref{lem:G->phi}. The second part follows from the fact that $G$ has $3m+2n+1=O(m+n)$ vertices, and assuming ETH, $3$-{\sc sat} cannot be solved in $2^{o(n+m)}$ time.\qed
\end{proof}

We now point out that $3$-{\sc svcfc} is \NP-complete for diameter-$d$ graphs for every fixed integer $d\ge 4$. Note that a similar reduction to the case $k\ge4$ colors in section~\ref{sec:k>=4} does not work for~$3$ colors. Instead, given an instance $\phi$ of $3$-{\sc sat}, we first construct an equivalent instance~$G$ for $3$-{\sc svcfc} in the proof of Theorem~\ref{thm:3333}. Observe that the distance in $G$ from the special vertex~$a$ to all other vertices $x\in V(G)\setminus\{a\}$ is at most~$2$. Then, given $d\ge 4$, let $G'$ be obtained from $G$ and $Q_n$, respectively, $R_n$ as follows; see Fig.~\ref{fig:QR} for the graphs $Q_n$ and $R_n$:
\begin{itemize}
\item If $d\ge 4$ is even, take $Q_n$ with $n=(d-2)/2$, and identify  the vertex $c_0$ in $Q_n$ with the special vertex~$a$ of $G$.
\item If $d\ge 5$ is odd, take $R_n$ with $n=(d-1)/2$, and identify the vertex $c_0$ in $R_n$ with the special vertice~$a$ of $G$.
\end{itemize}
\noindent
Then $\diam(G')=d$ (by construction and by Observation~\ref{obs:2}), and it is a routine to verify that $G$ is strongly conflict-free vertex-connection $3$-colorable if and only if $G'$ is. Moreover, as $d$ is a given constant, $|V(G')|=O(|V(G)|)$, Theorem~\ref{thm:3333} implies:
\begin{theorem}\label{thm:k=3}
$3$-{\sc svcfc} is \NP-complete even when restricted to diameter-$d$ graphs for every fixed $d\ge 3$, and assuming ETH, cannot be solved in $2^{o(n)}$ time on diameter-$d$ $n$-vertex graphs. 
\end{theorem}

\section{Two polynomial cases}\label{sec:P}
In this section we describe two classes of graphs with diameter at most three for which we are able to find an optimal strong conflict-free vertex-connection coloring in polynomial time. 

\subsection{Split graphs}
Let $G=(V,E)$ be a connected split graph with a partition $V=C\,\dot\cup\, I$ of the vertex set $V$ into a clique $C$ and an independent set $I$. Observe that the diameter of~$G$ is at most~$3$. 
To avoid triviality, we assume that $G$ is not a complete graph, in particular, $C$ is a non-empty clique and $I$ is a non-empty independent set. We may assume further that every vertex in $I$ is non-adjacent to a vertex in $C$: if some $u\in I$ is adjacent to all vertices in $C$ then replace $C$ by $C\cup\{v\}$ and $I$ by $I\setminus\{v\}$. 

If $G$ is a star then the unique $2$-coloring of $G$ is a (optimal) strong conflict-free vertex-connection coloring. So, let us assume that $G$ is not a star, in particular $|C|\ge2$. 
If $|C|=2$ then $G$ admits an optimal strong conflict-free vertex-connection $3$-coloring~$f$ with $f(C)=\{1,2\}$ and $f(I)=\{3\}$.  

So let $|C|\ge 3$. 
We remark that any optimal coloring of $G$ uses exactly $|C|$ colors and such a coloring need not be a strong conflict-free vertex-connection coloring of $G$. We will see that, however, $G$ always admits a strong conflict-free vertex-connection coloring with $|C|$ colors which can be computed efficiently.

Write $k=|C|\ge 3$, and color each vertex~$v\in C$ with an own color $f(v)\in [k]$. Then for each vertex $u\in I$, if some neighbor of $u$ is colored with color~$k$, color~$u$ with the smallest color $c$ that does not occur in $f(N(u))$. Otherwise color $u$ with color $c$, where $c\notin f(N(u))$ is the smallest upper bound of $f(N(u))$. 
More precisely, color $u$ with color $f(u)=c$, where
\[
c=\begin{cases}
\min\, [k]\setminus f(N(u)), & \text{if $k\in f(N(u))$}\\
\min \{i\in [k]\setminus f(N(u))\mid i>\max f(N(u))\}, & \text{otherwise}. 
\end{cases}
\]
Note that, as $u$ has a non-neighbor in $C$, the color $c$ as defined above always exists. By definition, the color $c$ does not occur in $f(N(u))$, hence $f$ is a proper $k$-coloring of $G$. 

\begin{lemma}
\label{lem:split}
$f$ is a strong conflict-free vertex-connection coloring. 
\end{lemma}
\begin{proof} Consider two vertices $u$ and $v$ at distance~$3$, and let $u-x-y-v$ be a shortest $u,v$-path. Note that $u, v\in I$, $x,y\in C$, and $N(u)\cap N(v)=\emptyset$. Let us assume that $f(u)=f(y)$ and $f(v)=f(x)$ (otherwise, $u-x-y-v$ is a strong conflict-free path for $u$ and $v$). 
Now, if $u$ has another neighbor $x'\not=x$ then $f(x')\not\in\{f(u), f(x)\}=\{f(v),f(y)\}$, hence $u-x'-y-v$ is a strong conflict-free path for $u$ and $v$. Similarly, if $v$ has another neighbor $y'\not=y$ then $u-x-y'-v$ is a strong conflict-free path for $u$ and $v$. 

Thus, it remains the case that $N(u)=\{x\}$ and $N(v)=\{y\}$. We argue that this case cannot occur: if $f(u)<f(x)$ then, by definition of $f(u)$, $f(x)=k$ and $f(u)=1$. Thus $f(y)=1$ and $f(v)=k$, hence, by definition of $f(v)$, $f(y)=k-1$. This implies $k=2$, a contradiction to $k\ge 3$. If $f(u)>f(x)$ then $f(v)<f(y)$ and the argument similar to the previous one yields the same contradiction. \qed
\end{proof}

Since $f$ uses $k=|C|$ colors, it is therefore optimal. Clearly, the coloring $f$ can be computed in polynomial time, hence we obtain: 

\begin{theorem}\label{thm:split}
An optimal strong conflict-free vertex-connection coloring of a given split graph can be computed in polynomial time.
\end{theorem}

\subsection{Co-bipartite graphs}
Let $G=(V,E)$ be the complement of a bipartite graph, that is, the vertex set $V=A\,\dot\cup\, B$ consists of two disjoint cliques $A$ and $B$. Observe that the diameter of~$G$ is at most~$3$. 

Let $E(A,B)$ be the edge set of $G$ with one end in $A$ and the other end in~$B$. To avoid triviality, we assume that $G$ is not a complete graph, in particular, $A$ and~$B$ are non-empty cliques. 
Note that $E(A,B)\not=\emptyset$ because all graphs considered are connected. 

\begin{lemma}
\label{lem:cobip1}
Assume that $|E(A,B)|=1$, and $|A|\ge |B|$. Then 
\begin{itemize}
\item[\em (a)] If $|A|=|B|$ then any optimal strong conflict-free vertex-connection coloring of $G$ uses $|A|+1$ colors and such a coloring can be computed in polynomial time.
\item[\em (b)] If $|A|>|B|$ then any optimal strong conflict-free vertex-connection coloring of $G$ uses $|A|$ colors and such a coloring can be computed in polynomial time.
\end{itemize}
\end{lemma}
\begin{proof}
Let $E(A,B)=\{xy\}$ with $x\in A$ and $y\in B$. Observe that the path between any vertex in $A$ and any vertex in $B$ is unique. Also observe that any coloring of~$G$ uses at least $|A|$ colors.

\smallskip\noindent
(a): Note that $|A|=|B|\ge 2$ because $G$ is not complete. Consider any coloring $f$ of $G$ with $|A|=|B|$ colors. Note that $f(x)\not=f(y)$, hence some vertex $u\in A\setminus\{x\}$ has color $f(u)=f(y)$ and some $v\in B\setminus\{y\}$ has color $f(v)=f(x)$. Thus, the shortest $u,v$-path in $G$ is not strong conflict-free under coloring~$f$. In particular, any strong conflict-free vertex-connection coloring of~$G$ must use at least $|A|+1$ colors. Now color $x$ with color $f(x)=|A|+1$ and properly color $(A\setminus\{x\})\cup B$ with colors $1, \ldots, |A|$. Then only vertex $x$ is colored with color~$|A|+1$, hence the coloring is a strong conflict-free vertex-connection coloring of~$G$ and is optimal.

\smallskip\noindent
(b): Color $A$ with $|A|$ colors $1, 2,\ldots, |A|$, where each vertex of $A$ has its own color. Let us assume that $x$ is colored by color~$1$, say.  Color~$B$ with $|B|$ colors in $\{2,\ldots,|A|\}$. Then only vertex $x$ is colored with color~$1$, hence the coloring is a strong conflict-free vertex-connection coloring of~$G$ and is optimal.\qed
\end{proof}

\begin{lemma}
\label{lem:cobip2}
Assume that $|E(A,B)|\ge 2$. Then any (vertex-)coloring of $G$ is a strong conflict-free vertex-connection coloring of $G$. 
\end{lemma}
\begin{proof}
Consider a coloring $f$ of $G$. Let $x\in A$ and $y\in B$ be two vertices of $G$ at distance~$3$. Then $N(x)\cap B=\emptyset$ and $N(y)\cap A=\emptyset$. Let $x-u-v-y$ be a shortest $x,y$-path with $u\in A$ and $v\in B$. If $f(x)\not= f(v)$ or $f(y)\not=f(u)$ then this path is strong conflict-free for $x$ and $y$. 
So let us consider the case $f(x)=f(v)$ and $f(y)=f(u)$. Now, since $|E(A,B)|\ge 2$, there is an edge $u'v'$ with $u'\in A$ and $v'\in B$; possibly $u'=u$ or $v'=v$ but not both. By symmetry, let $u'\not=u$. Then $f(u')\not\in\{f(x),f(u)\}=\{f(v),f(y)\}$ because $A$ is a clique and $f$ is a proper coloring of $G$. Thus, $x-u'-v'-y$ is strong conflict-free path connecting $x$ and~$y$. Hence $f$ is a strong conflict-free vertex-connection coloring of $G$.\qed 
\end{proof}

Note that a partition into two disjoint cliques, as well as an optimal coloring of a given co-bipartite graph can be computed in polynomial time. Hence, with Lemmata~\ref{lem:cobip1} and~\ref{lem:cobip2}, we immediately obtain:
\begin{theorem}\label{thm:cobip}
An optimal strong conflict-free vertex-connection coloring of a given co-bipartite graph can be computed in polynomial time.
\end{theorem}

\section{Conclusions}\label{sec:con}
In this paper we initiate the study of the computational complexity of strong conflict-free vertex-connection colorings, a new variant of graph coloring. We provide the first hardness results and point out some polynomial cases. 
The following open problems are immediate from our study:
\begin{itemize}
\item[(1)] 
Is there a class of graphs in which computing the strong conflict-free vertex-connection number is polynomially but computing the chromatic number is hard? The results obtained in this paper indicate that such a graph class may not exist.
\item[(2)] 
Characterize and recognize connected graphs (of diameter at least three) in which every proper vertex-coloring is a strong conflict-free vertex-connection coloring. A good characterization for co-bipartite graphs and split graphs with this property could be obtained from the discussion in section~\ref{sec:P}.
\item[(3)] 
Carl Feghali (private communication) observes that the proof of the upper bound $\svcfc(T)\le \lceil\log_2(n+1)\rceil$ for $n$-vertex trees~$T$ in~\cite{LiW18+} can be turned into a polynomial-time algorithm solving \textsc{svcfc} when restricted to trees. Thus, it is natural to extend the tractable cases of trees and split graphs to large classes. In particular, is \textsc{svcfc} solvable in polynomial time when restricted to classical graph classes such as chordal graphs, bipartite graphs or planar graphs?
\end{itemize}

\bibliographystyle{plainurl} 
\bibliography{svcfc.bib}

\end{document}